%% file: ecc2016.tex
\documentclass[letterpaper,10pt,conference]{ieeeconf}  

\IEEEoverridecommandlockouts                              

\input{include.tex}

\title{\LARGE \bf
 	Reducing complexity of autonomous control agents for verifiability
 }

\author{Paolo Izzo$^{1}$, Hongyang Qu$^{2}$ and Sandor M. Veres$^{3}$
\thanks{*This work was supported by Thales UK and The University of Sheffield }
\thanks{$^{1}$Paolo Izzo is a PhD student at ACSE, The University of Sheffield {\tt \small pizzo1@sheffield.ac.uk}}%
\thanks{$^{2}$ Dr Hongyang Qu is Senior Research Fellow at ACSE, The University of Sheffield,{\tt \small h.qu@sheffield.ac.uk}}%
\thanks{$^{3}$Sandor M. Veres is a Professor of Autonomous Systems at ACSE, The University of Sheffield, {\tt\small s.veres@sheffield.ac.uk}}%
}

\begin{document}

\maketitle
\thispagestyle{empty}
\pagestyle{empty}

\begin{abstract}

The AgentSpeak type of languages are considered for decision making in autonomous control systems. To reduce the complexity and increase the verifiability of decision making, a limited instruction set agent (LISA) is introduced. The new decision method is structurally simpler than its predecessors and easily lends itself to both design time and runtime verification methods. The process of converting a control agent in LISA into a model in a probabilistic model checker is described. Due to the practical complexity of design time verification the feasibility of runtime probabilistic verification is investigated and illustrated in the LISA agent programming system for verifying symbolic plans of the agent using a probabilistic model checker. 

\end{abstract}

\section{Introduction}
\label{sec:intro}



In control sciences the concept of autonomous control has emerged as an upgrade of feedback control, where the controller also decides what and how to control  in terms of objects, parameters and performance \cite{astrom92,astrom2010}. 

Attempts towards autonomous decision-making software were initially made through \gls{oop}. A first example of \gls{oop}, designed with decision-making in mind, is probably  given by Ng and Luk in \cite{ng1995}. A more recent example of research in this direction can be found in \cite{ridao2002}, where Ridao presents a layered \gls{oop} control architecture, with deliberative, control execution and reactive layers. For autonomous control, the \gls{oop} framework is usually linked to \glspl{hs} modelling software. A few examples   applications can be found in \cite{balarin2002,pinto2005,silva2001}.


As a next stage of autonomous decision making, formal description of autonomous agents can be found in \cite{veres2011,wooldridge2002,wooldridge1995}. Most architectures are structured in a layered way, as described in references \cite{gat1998,alami1998,simmons1998}, namely a \emph{Deliberator layer}, a \emph{Sequencer layer} and a \emph{Controller layer}, with different levels of interaction between the layers according to the structure of the system. An exception against this trend that is worth mentioning is the CLARAty architecture \cite{nesnas2006,volpe2001} developed by \acrshort{nasa} where the Deliberator layer and the Sequencer layer are merged together in continuous replanning schemes.

To beat the complexity of decision making which arose, one of the earliest and best known {\it behaviour-based architectures} was Brook's \emph{subsumption architecture} \cite{brooks1990}. An interesting recent application of this kind of architecture is the MOOS-IvP project \cite{moosivpwebsite,benjamin2012}, which is mostly implemented for \glspl{uuv} and \glspl{usv}. \gls{moos} \cite{newman2003} is an autonomy \emph{middleware}, structured in a star-like fashion. Every application only interfaces with a central database, and the inter-process communication between applications happens with a Publish/Subscribe policy. \gls{ivp} is a \gls{moos} application that implements a behaviour-based decision-making engine \cite{benjamin2004}. 
At the end of every reasoning cycle, the engine solves a multi-objective optimisation problem over all the objective functions generated by the behaviours.

Another popular ''anthropomorphic'' approach to the implementation of autonomous agents is the \gls{bdi} architecture which are implemented in programming \cite{bordini2007,veres2011}. \gls{bdi} agent architectures are characterised by three large abstract sets: \emph{Beliefs}, \emph{Desires} and \emph{Intentions}. 
The beliefs set represents the information the agent has about the world, the desires set represents something that the agent \emph{might} want to accomplish and the intentions set represents the set of options that the agent is committed to work towards. The most known implementations of the \gls{bdi} architecture are the \gls{prs} \cite{georgeff1986,georgeff1987} and \emph{AgentSpeak} \cite{rao1996}. AgentSpeak, and in particular \emph{Jason} \cite{bordini2007,jasonmanual} and Jade\cite{bellifemine2000,bellifemine2007}, fully embrace the philosophy of \gls{aop} \cite{shoham1993}, offering a Java based interpreter that can be customised according to the designer needs. 

In  this paper  we aim to analyse the Jason reasoning cycle to outline its design disadvantages for verification complexity and propose a new architecture called Limited Instruction Set Agent (LISA), which is based on previous expansions of AgentSpeak such as Jason and Jade, while relying more on external planning processes, abstractions from planning and optimisation for decision making called by the agent. We  also propose to use model checking techniques at
two levels: design-time   and at run-time. For the design time model checking we prove that LISA is  implementable as a \gls{dtmc} and that probabilistic model checking is applicable.
Given the lack of definition of the plan selection function in Jason, our idea is to improve the architecture with a run-time probabilistic model checking by predicting the outcome of applicable plans and actions.

The remaining sections of the paper provide an abstract model for the reasoning cycle in Jason and {LISA}, abstractions to DTMC, principles of applicable plan selection, how to use model checking runtime  and conclusions complete the paper. 

\section{The agent reasoning cycle}
\label{sec:background}


By analogy to previous definitions \cite{lincoln2013,wooldridge2002,veres2011} of AgentSpeak-like architectures, we define our agents as a tuple:
\begin{equation}
\label{eq:agent}
\mathcal{R}=\{ \mathcal{F},B,L,\Pi,A\}
\end{equation}
where:
\begin{itemize}
\item 
	$\mathcal{F} = \{p_1,p_2,\ldots,p_{n_p}\}$ is the set of all predicates.
\item
	$B \subset \mathcal{F}$ is the total atomic belief set. The current belief base at time $t$ is defined as $B_t \subset B$. At time $t$ beliefs that are added, deleted or modified are considered \emph{events} and are included in the set $E_t \subset B$, which is called the \emph{Event set}. Events can be either \emph{internal} or \emph{external} depending on whether they are generated from an internal action, in which case are referred to as ``mental notes'', or an external input, in which case are called ``percepts''. 
\item
	$L = \{l_1,l_2,\ldots\,l_{n_l}\}$ is a set of logic-based implication rules.
\item
	$\Pi = \{\pi_1,\pi_2,\ldots,\pi_{n_\pi}\}$ is the set of executable plans or \emph{plans library}. Current applicable plans at time $t$ are part of the subset $\Pi_t \subset \Pi$, this set is also named the \emph{Desire set}. A set $I \subset \Pi$ of intentions is also defined, which contains plans that the agent is committed to execute. 
\item 
	$A = \{a_1,a_2,\ldots,a_{n_a}\} \subset \mathcal{F} \setminus B$ is a set of all available actions. Actions can be either \emph{internal}, when they modify the belief base or generate internal events, or \emph{external}, when they are linked to external functions that operate in the environment.
\end{itemize}


 AgentSpeak like languages, including LISA,  can be fully defined and implemented by listing the following characteristics:
\begin{itemize}
\item \emph{Initial Beliefs}.\\
	The initial beliefs and goals $B_0 \subset F$ are a set of literals that are automatically copied into the \emph{belief base} $B_t$ (that is the set of current beliefs) when the agent mind is first run.
\item \emph{Initial Actions}.\\
	The initial actions $A_0 \subset A$ are a set of actions that are executed when the agent mind is first run. The actions are generally goals that activate specific plans.
\item \emph{Logic rules}.\\
	A set of logic based implication rules $L$ describes \emph{theoretical} reasoning to improve the agent current knowledge about the world.
\item \emph{Executable plans}.\\
	A set of \emph{executable plans} or \emph{plan library} $\Pi$. Each plan $\pi_j$ is described in the form:
	\begin{equation}
		p_j : c_j \leftarrow a_1, a_2, \ldots, a_{n_j}
	\end{equation}
	where $p_j \in B$ is a \emph{triggering predicate}, which allows the plan to be retrieved from the plan library whenever it comes true, $c_j \in B$ is called the \emph{context}, which allows the agent to check the state of the world, described by the current belief set $B_t$, before applying a particular plan, and $a_1, a_2, \ldots, a_{n_j} \in A$ is a list of actions.
\end{itemize}

\begin{figure*}
	\centering
	\includegraphics[width=180mm]{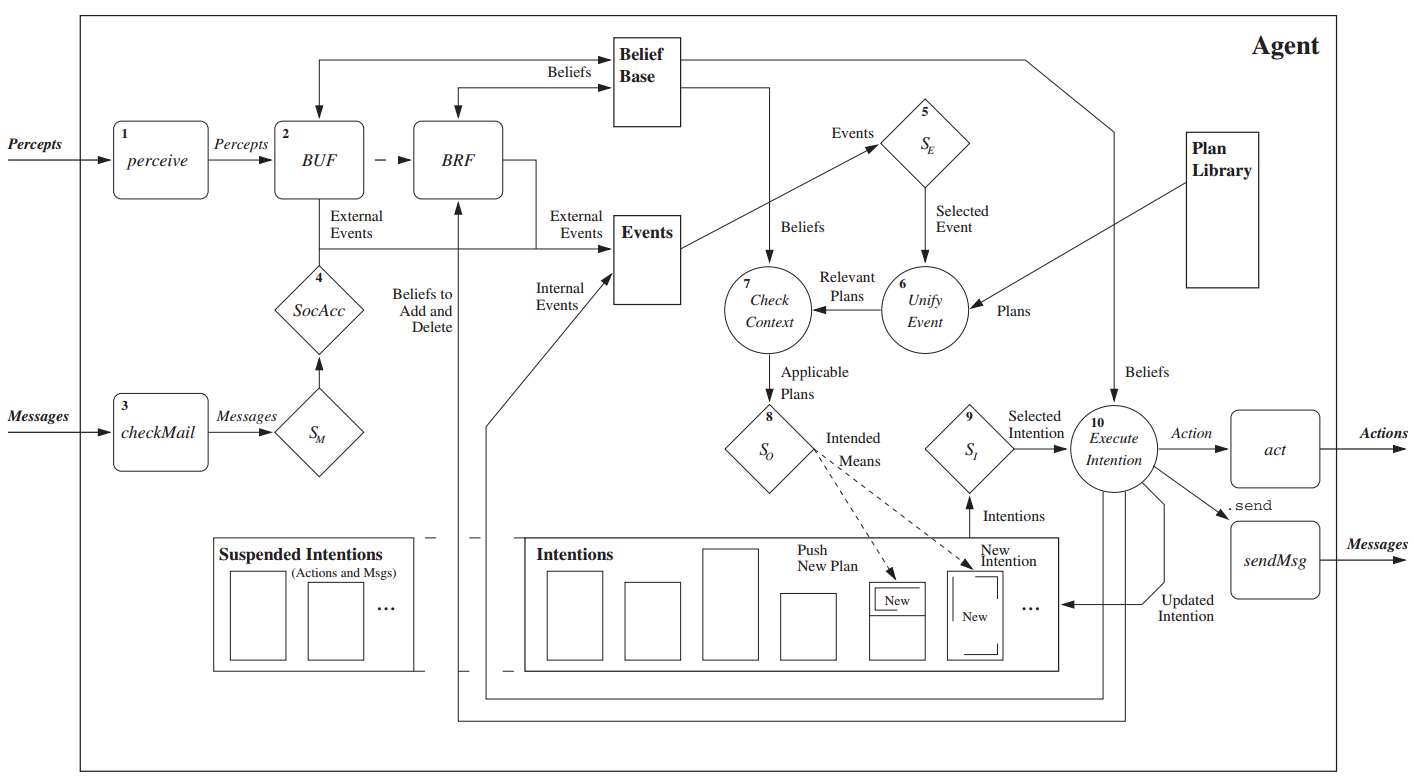}
	\caption{The \emph{Jason} reasoning cycle \textnormal{\cite{bordini2007}}. LISA simplifies this scheme while anabling limited complexity reasoning about future actions by model checking. }
	\label{fig:jason2}
\end{figure*}

For comparison we review the semantics of the reasoning cycle in Jason (see Fig. \ref{fig:jason2} and \cite{jasonmanual,bordini2007}) and similar languages and then we will introduce and build on the simplifications by LISA in the next section. The reasoning cycle of Jason can be summarised by four main steps:

\begin{enumerate}
\item
	\emph{Belief base update}.\\ The agent updates the belief base by retrieving information about the world through perception and communication. The job is done by two functions called \emph{Belief Update Function (BUF)} and \emph{Belief Review Function (BRF)}. The BUF takes care of adding and removing beliefs from the belief base; the BRF updates the set of current events $E_t$ by looking at the changes in the belief base.
\item
	\emph{Trigger Event Selection}.\\ For every reasoning cycle, only a single event can be dealt with. For this reason a function called \emph{Event Selection Function} selects one of the events from the current event set $E_t$:
	\begin{equation}
		S_E : \wp(B) \rightarrow E_t
	\end{equation}
	where $\wp(\cdot)$ is the so called \emph{power operator} and represents the set of all possible subset of a particular set. We will call the current selected event $S_E(E_t)=e_t$. 
\item
	\emph{Plan Selection}.\\ Once the triggering event to be dealt with in the current reasoning cycle is selected, the agent retrieves from the plan library a set of all executable plans that feature $e_t$ (the current event) as triggering event; then these plans are checked for compatible context. All the plan that meet the triggering event and the context are selected as part of an \emph{Applicable Plans} set $\Pi_t$, also called \emph{Desire set}. From the set of applicable plans a function called \emph{Applicable Plan Selection Function} $S_O$ (``O'' stands for Option) selects the plan that will actually be used to pursue the selected goal or deal with the selected event:
	\begin{equation}
		S_O: \wp(\Pi) \rightarrow \Pi
	\end{equation}
	We will call the current selected plan $S_O(\Pi_t)=\pi_t$.
\item
	\emph{Intention Selection}.\\ 
	Once a plan has been selected, the agent is committed to execute and complete it, in order to achieve its goals. A plan that is selected for execution is called \emph{intention} and it is copied in the \emph{Intentions Set} $I$. At every reasoning cycle the agent is only able to execute a single action; this means that the intention base is likely to contain multiple plans at any given time. A function called \emph{intention selection function} $X_I$ selects the plan to be executed in the current reasoning cycle:
	\begin{equation}
		X_I : \wp(\Pi) \rightarrow \Pi
	\end{equation}
	The selected intention $X_I(I)=\pi^{(i)}_t$ is then taken on for execution. 
\item 	\emph{Action Execution}.\\	
	The agent takes the next action from the plan and call external or internal functions to execute it. Once the action is executed it is removed from the Intentions Set. 
\end{enumerate}

The Jason agent architecture does not provide for any articulated implementation of $S_E$, $S_O$ and $X_I$. Furthermore the agent is restricted to deal with a single event per reasoning cycle, which does not give any reasonable computational advantage, and it is arguably an unnecessary limitation considering the accessibility of modern processing units. For these reasons and the complexity of using intention sets, we propose a new simplified architecture called LISA while keeping essential BDI features of programming.


\section{Reasoning cycle in LISA}
\label{sec:lisa}

The purpose of simplifying the reasoning structures relative to AgenSpeak languages in use today, is threefold:

\begin{enumerate}
\item To make complexity of design time verification simpler
\item To enable fitting of model checking for runtime verification of selected plans and actions.
\item To make logic based decisions more robust to mistakes in the a priori design of agent reasoning 
and especially less fragile to badly timed appearance of predicates in the belief base. 
 \end{enumerate}
 
With regards to the definition given in Equation \ref{eq:agent} we make the following distinctions:
\begin{itemize}
\item \emph{Beliefs and Goals}.\\
	In the Jason agent there is a distinction between beliefs and goals. Beliefs are what the agent knows about the world, and goals are special mental notes that the agent does not keep in the belief base when the plan they trigger is achieved. This distinction in a practical sense does not improve the implementation process. For this reason in LISA we drop the distinction between beliefs and goals.
\item \emph{External Actions}.\\
	In LISA we introduce a new classification for external actions that can be either of type \emph{runOnce}, if the action terminates itself after a single execution or \emph{runRepeated}, if the action runs continuously until actively stopped by the agent. The latter implies that the agent is capable of monitoring the outcome of the \emph{runRepeated} actions using its perception processes. In either case the external functions that execute actions can send predicates to the current belief base in the form of \emph{action feedbacks}.
\item \emph{Perception}.\\
	In LISA perception predicates can be of three types: \emph{sensory perception}, \emph{communication} and \emph{action feedbacks}. As in Jason all the perception predicates generate events that can in turn trigger plans that are part of the plan library.
\item \emph{Logic rules}.\\
	In Jason the only way to internally add or remove beliefs from the belief base is with  internal actions. In LISA predicates also arrive from action processes directly.   Logic-based implication rules are not deeply implemented, and the main text itself \cite{bordini2007} advises against their use. In LISA we process logic rules until stable conclusions in each reasoning cycle. 
\end{itemize}

We now present the LISA agent reasoning cycle, highlighting the differences with the Jason reasoning cycle presented in Section \ref{sec:background}. The simplified reasoning cycle is illustrated in Fig. \ref{fig:lisacycle}. Subscript $t$ indicates the indexing of the reasoning cycle.

\begin{enumerate}
\item
	\emph{LISA belief base update}.\\ 
	At the beginning of every reasoning cycle, the BUF checks for all the input coming from perception, action feedback and communications and updates the current belief set $B_t$. The BUF function also removes perception and action feedback predicates if they are not received persistently. For instance indication of a task completed may be a feedback sent only once while sensor based detection of dynamical instability can be a persistent feedback. The BRF generates a set of events $E_t=B_t\setminus B_{t-1}$ at every reasoning cycle. The LISA agent does not select a single trigger for intention but executes all plans with true context  plans in a multi-threaded way.
\item
	\emph{Plan selection}.\\
	Next the agent looks at the current event set $E_t$ and retrieves all the plans from the plan library $\Pi$ that are triggered by these events, it checks that the context meets the current beliefs, and copies the plans to the Desire set $\Pi_t \subset \Pi$.
\item
	\emph{Action Execution}.\\
	The agents retrieves the next action to be executed from each plan  and calls for externals or internal functions to execute the action. Similarly to Jason, a plan is suspended while waiting for a completion feedback from an action.
\end{enumerate}

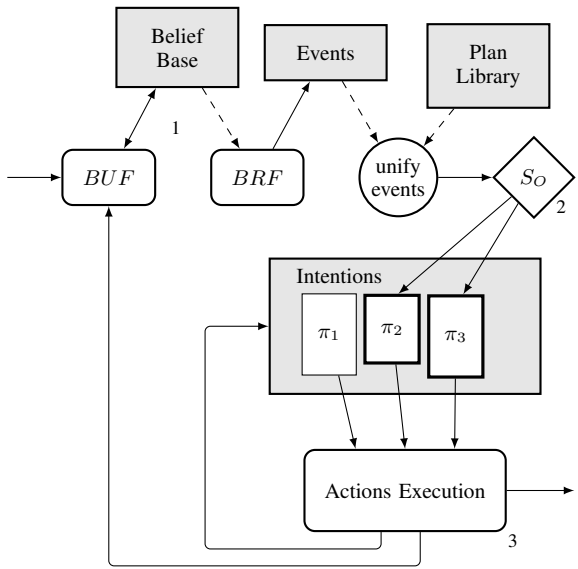
\begin{figure}[htbp]
	\centering

\begin{tikzpicture}[auto,node distance=8mm,>=latex,font=\small,scale=0.9, every node/.style={transform shape}]

\tikzstyle{rect}=[thick,draw=black,inner sep=3mm,fill=white!90!black,text width=12mm,text centered]
\tikzstyle{diamond}=[thick,draw=black,minimum height=12mm,minimum width=12mm,shape=diamond]
\tikzstyle{roundcorners}=[thick,draw=black,rounded corners,inner sep=3mm]
\tikzstyle{round}=[thick,draw=black,inner sep=0mm,text width=10mm, text centered,shape=circle]

	\node (ref1) at (0,0) []{};
	\node (buf) [roundcorners,right= of ref1] {$BUF$};
	\node (brf) [roundcorners,right= of buf] {$BRF$};
	\node (bb) [rect,above right=9mm and -6mm of buf] {Belief Base};
	\node (ue) [round,right= of brf] {unify events};
	\node (events) [rect,above right=10mm and -6mm of brf] {Events};
	\node (So) [diamond,right= of ue] {$S_O$};
	\node (plans) [rect,above right=6mm and 0mm of ue] {Plan Library};
	
	\node (inten) at (60mm,-22mm) [rect,minimum height=20mm,minimum width=40mm,text width=20mm] {};
	\node (inten2) [above left=-5mm and -18mm of inten]{Intentions};
	\node (plan1) [fill=white,draw=black,minimum height=12mm,minimum width=8mm,below left=-15mm and -13mm of inten]{$\pi_1$};
	\node (plan2) [fill=white,draw=black,very thick,minimum height=10mm,minimum width=8mm,above right=-10.5mm and 1mm of plan1]{$\pi_2$};
	\node (plan3) [fill=white,draw=black,very thick,minimum height=12mm,minimum width=8mm,above right=-12.5mm and 1mm of plan2]{$\pi_3$};
	
	\node (act) [minimum height=12mm,roundcorners,below= of inten] {Actions Execution};
	
	\node[above right=1mm and 1mm of buf] (n1) {\footnotesize 1};
	\node[below right=-1mm and -1mm of So] (n2) {\footnotesize 2};
	\node[below right=-1mm and -1mm of act] (n3) {\footnotesize 3};
	
	\draw[->] (ref1) -- (buf);
	\draw[<->] (bb) -- (buf);
	\draw[->,dashed] (bb) -- (brf);
	\draw[->] (brf) -- (events);
	\draw[dashed,->] (events) -- (ue);
	\draw[->] (ue) -- (So);
	\draw[dashed,->] (plans) -- (ue);
	\draw[->] (So) -- (plan3.80);
	\draw[->] (So) -- (plan2.80);
	\draw[->] (plan1) -- (act.140);
	\draw[->] (plan2) -- (act.90);
	\draw[->] (plan3) -- (act.40); 
	\draw[->,rounded corners=2] (act.-120) |- +(-2.5mm,-2.5mm) -| +(-26mm,20mm) |- (inten);
	\draw[->,rounded corners=2] (act.-70) |- +(-5mm,-5mm) -| (buf);
	\draw[->] (act) -- +(25mm,0mm);
	
\end{tikzpicture}
	\caption{Reasoning cycle of LISA: the plans run in a multi-threaded way, avoiding the need of $S_E$ and $S_I$}
	\label{fig:lisacycle}
\end{figure}
The LISA agent reasoning cycle reduces complexity compared to previous
implementations and the agent decision making process can be easily
modelled as a DTMC, with particular attention to the selection of the executable plans, as proven in the next section. 


\section{LISA abstractions to DTMC}
\label{sec:dtmcabstraction}

This section will address the problem of modelling and model checking the decision processes of the LISA architecture as mathematically described in the previous section. 
Model checking is understood here before the application of the agent in the environment, i.e. at design time. 
The objective is to verify the functionality of the agent code in the environment, assuming the software has been verified to precisely deliver its designed functionality in LISA. 
The functionality of LISA is to be verified against a series of queries using probabilistic model checking techniques. 
In this paper, we will be using PRISM \cite{kwiatkowska2011,prismwebsite}, a popular probabilistic model checker, to perform the verification. Detailed study of the list of verification queries in PRISM goes beyond the scope of this paper. Here our objective is to consider the specification itself and show that verification can in principle be carried out. 

We will assume here that the LISA agent to be verified will function in a physical environment which it needs to perceive and model using sensing instruments such as IMUs, Lidars, Sonars, cameras and so on. 
While it plans and takes actions, the agent will not always
completely succeed with its intentions. It is assumed however that any
action process is able to assess its own outcome, i.e. success, partial success or failure, 
and that it is able to feed this information back to the belief base of the agent. 
The action feedbacks can potentially include direct information on the cause of problems in case of failure or partial failure, however the level of detail will depend on the amount of the agent programmers' efforts to
identify what is necessary to achieve a useful intelligent machine
behaviour. This way the decision making of the agent can be gradually
upgraded by more and more detailed analysis of its own work.  For our
verification theory the level of depth in the action-feedback is not relevant
but its existence is.

The agent and the environment is considered as a single system to be verified. In fact there will be three subsystems to be modelled for verification purposes:  (1) the robotic agent (2) the physical environment (3) the human operator (see Fig. \ref{fig:fullmodel}).
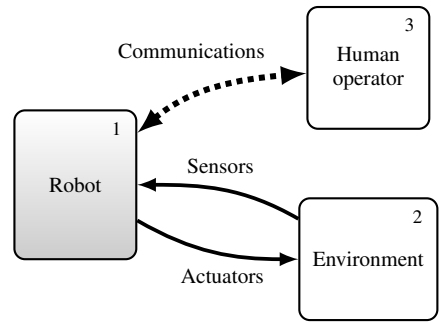
\begin{figure}[htbp]
	\centering

\begin{tikzpicture}[auto,node distance=10mm,>=latex,font=\small,scale=0.9, every node/.style={transform shape}]

	\tikzstyle{rc}=[thick,draw=black,inner sep=3mm,text width=12mm,text centered,rounded corners,minimum height=18mm]
	\tikzstyle{num}=[font=\footnotesize]
	
	\node[rc,left color=white,right color=white!80!black,shading angle=0,minimum height=22mm] (robot) {Robot};
	\node[rc,above right=-3mm and 25mm of robot] (human) {Human operator};
	\node[rc,text width=,inner sep=2mm,below= of human] (environment) {Environment};
	
	\node[num,above right=-5mm and -5mm of robot] (n1) {1};
	\node[num,above right=-5mm and -5mm of environment] (n2) {2};
	\node[num,above right=-5mm and -5mm of human] (n3) {3};
	
	\draw[<->,line width=0.8mm,dotted] (robot) [out=40,in=185] to node[text width=10mm,above left=4mm]{Communications} (human);
	\draw[->,line width=0.5mm] (environment) [out=150,in=0] to node[text width=10mm,above=1mm]{Sensors} (robot);
	\draw[->,line width=0.5mm] (robot) [out=-30,in=180] to node[text width=10mm,below=1mm]{Actuators} (environment);

\end{tikzpicture}
	\caption{Subsystems of the full model. The human operator can send mission objectives to the Robot, and receive reasoning feedback. The robot interacts with the environment with an array of sensors and actuators.}
	\label{fig:fullmodel}
\end{figure}
In a broader approach to verification, when we verify the agent for a broader class of environmental possibilities, the action feedbacks are events with random characteristics which do not reflect the internal structure of the environment. 

\begin{theorem}
\label{th1}
Assuming an independent probability distribution of all action
feedbacks and sensory events from the environment, the complete
decision making of a LISA can be modelled by a \gls{dtmc}.
\end{theorem}
\begin{proof} (Outline)
   As during the execution of any action(s) in some plan(s) is(are) represented by specific predicates in $B$, the power set $\wp(B)$ (the set of subsets of $B$) forms a state space for the agent as any $s\in \wp(B)$ provides complete information for the decision making of the agent. This holds true regardless of the fact that the memory of the agent (in terms of its temporal world model) influences its decisions, as those decisions are modelled by probabilistic predicate feedback from processes of the agent analysing its external world. The state of the \gls{dtmc} is initialised by its initial set of beliefs, goals and probabilities of feedback from its initial actions. Transition to a new state of the agent only happens at the end of each reasoning cycle. The new set of predicates for the new state is developed in two steps: (1) by deterministic application of reasoning rules, triggering or closure of new plans and elimination of predicates due to the semantics of LISA's execution (2) probabilistic appearance of perception predicates and action feedback predicates in $B$.  The probability distribution of new sensory and action feedback predicates is well defined in any LISA program and hence modelling as a \gls{dtmc} can be completed. 
\end{proof} 

The fact that we do not need a MDP (Markov Decision Process) to model the LISA agent and that a \gls{dtmc} suffices, signifies the importance of defining limited instruction set agent for verification purposes to reduce complexity. Nevertheless, a lot can be done to further reduce complexity of design-time verification  through  programming style in LISA. 

Another approach to environmental modelling is to abstract away the simulation of the environment so that perception events and actions feedbacks are modelled by probabilistic dependency in terms of \gls{dtmc}s. Similarly, the behaviour of human operators can also be formulated in terms of a \gls{dtmc}. The next formal result covers this case. 

\begin{theorem}
\label{th2}
Assuming that both human communication and the physical environment are possible to model by \gls{dtmc}s, which are possibly inter-dependent by conditional probabilities, the complete decision making of a LISA in its environment can be modelled by a \gls{dtmc}. 
\end{theorem}
\begin{proof} (Outline)
 The alteration relative to the proof of Theorem 1 is that the probabilistic appearance of perception  and action feedback predicates is this time determined by the \gls{dtmc}s of human communication and interaction with the agent and predicates of action feedback from the physical environment.  Due to the \gls{dtmc}s for both human responses as well as environmental responses well defined,  all the state transition probabilities can be calculated by basic rules of conditional probabilities and the proof is hence completed.   
\end{proof} 

It directly follows from Theorem \ref{th1} and \ref{th2} that the decision making system of a LISA agent can be verified in PRISM at design time.  The following sections will examine whether executable plan set selection can be done in terms of \emph{runtime model checking}, which can also be termed as \emph{runtime verification of executable plan selection} by the agent.

\section{Principles of applicable plan selection}
\label{sec:principles}

In this section we discuss the general definition of a more advanced method for the Applicable Plan Selection Function $S_O$.

In order to improve the agent's  understanding of the world in association with its own internal state, a finite set of so called ``operational states'' is defined: 
\begin{equation}
	X = \left\{ x_1,x_2,\ldots,x_{n_x} \right\} \subset \mathcal{F}
\label{eq:opstates}
\end{equation}
Operational states are beliefs, predicates of $\mathcal{F}$, that represent high level states of the agent, for example ``Waiting for instructions'' or ``Exploring area of interest''. Operational states are not mutually exclusive, more than one can be active at any given time. We will call $X_t \subset X$ the set of active operational states at time $t$.\\
The agent has to take care of updating its own operational state by applying a set of logic implication rules. These rules are based on statements called ``pre-conditions''.  
Pre-conditions include beliefs generated by the percept process as well as mental notes created during plans executions. \\
The associations between plans, that are effectively sequences of actions that the agent can take and their possible outcomes, 
can mostly be made in advance, especially using advanced simulation software. However for advanced systems such as a fully operational autonomous vehicles, the amount of {\it possible outcomes} can easily become too  high to be accounted for at runtime. In order to limit the number of associations, we apply the following classes of constraints:
\begin{itemize}
\item \emph{Temporal order of actions}.\\
	The agent can be programmed so it has the ability to memorise and later check for past actions. Some actions must be preceded by others, for instance the agent is not allowed to pick up an instrument that it never deployed in the environment.
\item \emph{Actions performed in the context of operational modes}.\\
	Operational modes allow to give a general context to the status of the mission and can be used to reduce the number of associations between course of actions and possible outcomes .
\item \emph{Environmental dependency}.\\
	An autonomous vehicle is designed to operate autonomously in any condition. External events shape the way the system is going to act throughout the mission. For instance if the agent experiences a communication loss with other agents, whom it was supposed to communicate with, the mission objectives may not be achievable any more and the agent may want to do something else instead.
\end{itemize}

Let us define an ``implication function'' $f_\mathcal{I}$ that associates every plan a finite set of possible outcomes, and the likelihood that each of those event would happen:

\begin{equation}
	\mathcal{I} : \Pi \rightarrow \wp \left( \wp(B) \times [0,1] \right)
\end{equation}
%

Once all the associations have been made, the agent can generate a graph that represents all the possible course of actions that the agent can go through according to its symbolic plan list for the future, using outcomes (events) associated with actions in the code of the plans.  Fig. \ref{fig:tree1} shows a simple two time steps representation of this concept.

\begin{figure}[htbp]
	\centering
	
\begin{tikzpicture}[auto, node distance=8mm,>=latex,scale=0.9, every node/.style={transform shape}]

	\tikzstyle{roundcorners}=[minimum height=8mm,minimum width=12mm,thick,draw=black,rounded corners=3]
	\tikzstyle{round}=[thick,draw=black,circle]

	\node (et1) [round] {$E_1$};
	\node (pi1) [roundcorners,below left=6mm and 0 of et1] {$\pi$};
	\node (pi12) [roundcorners,below right=6mm and 0 of et1] {$\pi$};
	\node (et21) [round,below right=8mm and 0mm of pi1] {$E_2$};
	\node (et22) [round,left= of et21] {$E_2$};
	\node (et23) [round,right= of et21] {$E_2$};
	\node (ref1) [minimum height=8mm,left=4mm of et22] {};
	\node (pi2) [roundcorners,below= of ref1] {$\pi$};
	\node (pi3) [roundcorners,right=5mm of pi2] {$\pi$};
	\node (pi4) [roundcorners,right=5mm of pi3] {$\pi$};
	\node (pi5) [roundcorners,right=5mm of pi4] {$\pi$};
	\node (et31) [round,below= of pi2] {$E_3$};
	\node (et32) [round,below= of pi3] {$E_3$};
	\node (et33) [round,below= of pi4] {$E_3$};
	\node (et34) [round,below= of pi5] {$E_3$};
	
	\node (legend) at (30mm,-3mm) [
		minimum height=2mm, 
		minimum width=20mm,
		thin,
		draw=black,
		left color=white,
		right color=black
	]{};
	\node (like) [above=.1mm of legend] {\small Likelihood};
	\node (n05) [below=.1mm of legend] {\footnotesize $0.5$};
	\node (n0) [left=4mm of n05] {\footnotesize $0$};
	\node (n1) [right=4mm of n05] {\footnotesize $1$};
	
	\draw[dotted] (et1) -- (pi1);
	\draw[dotted] (et1) -- (pi12);
	
	\draw[white!50!black,thick,->] (pi1) -- (et21);
	\draw[white!40!black,thick,->] (pi1) -- (et22);
	\draw[white!10!black,thick,->] (pi1) -- (et23);
	\draw[white!80!black,thick,->] (pi12) -- (et21);
	\draw[white!20!black,thick,->] (pi12) -- (et23);
	
	\draw[dotted] (et21) -- (pi3);
	\draw[dotted] (et21) -- (pi4);
	\draw[dotted] (et21) -- (pi5);
	\draw[dotted] (et22) -- (pi2);
	\draw[dotted] (et22) -- (pi3);
	\draw[dotted] (et23) -- (pi4);
	\draw[dotted] (et23) -- (pi5);
	\draw[white!20!black,thick,->] (pi2) -- (et31);
	\draw[white!80!black,thick,->] (pi2) -- (et32);
	\draw[white!60!black,thick,->] (pi3) -- (et31);
	\draw[white!40!black,thick,->] (pi3) -- (et32);
	\draw[white!10!black,thick,->] (pi4) -- (et32);
	\draw[white!90!black,thick,->] (pi4) -- (et33);
	\draw[black,thick,->] (pi5) -- (et34);
	
	\draw[dotted] (et31) -- +(-7mm,-10mm);
	\draw[dotted] (et31) -- +(3mm,-10mm);
	\draw[dotted] (et32) -- +(-10mm,-10mm);
	\draw[dotted] (et32) -- +(5mm,-10mm);
	\draw[dotted] (et33) -- +(-7mm,-10mm);
	\draw[dotted] (et33) -- +(3mm,-10mm);
	\draw[dotted] (et34) -- +(-10mm,-10mm);
	\draw[dotted] (et34) -- +(5mm,-10mm);
	
	\node[below left=-5mm and 10mm of pi1,text width=20mm] (ar1) {\footnotesize\emph{random outcome}};
	\draw[thin] (ar1) -- +(21mm,-6mm);
	\node[below left=-7mm and 6mm of et22,text width=20mm,text centered] (ar2) {\footnotesize\emph{triggered\\[-1mm] plan selection}};
	\draw[thin] (ar2.-60) -- +(11.5mm,-3mm);
	
\end{tikzpicture}
	
	\caption{Example of possible course of plans tree: every plan can generate one or more events set with a certain probability, illustrated by the colour of the arrow. Every set of events in turn can trigger more than one plan.}
	\label{fig:tree1}
\end{figure}
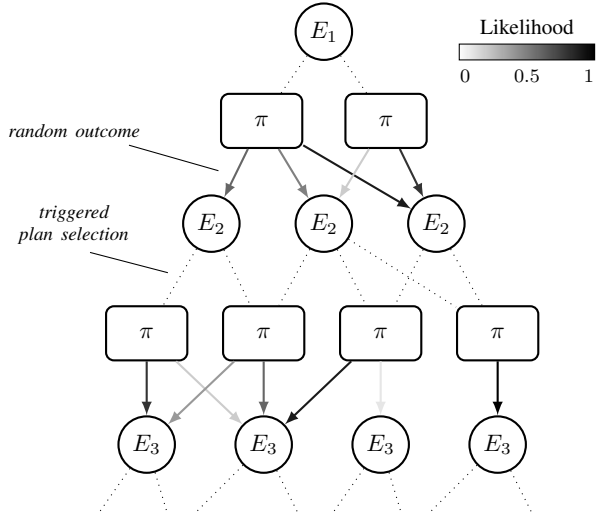

Let us define a ``course of plans'' as a sequence of plans $\hat{\Pi} \subset \Pi$, a branch of the tree in Fig. \ref{fig:tree1}, associated with a particular goal represented by a set of events $E_t$ that the agent is committed to pursue, and the likelihood that the course of action leads to the successful achievement of said goal:

\begin{equation}
	c_\pi^{(j)} = \left\{ \hat{\Pi}^{(j)}, E_t^{(j)}, \lambda^{(j)} \right\} 
\end{equation}
where the index $j$ numbers each branch of the three in Fig. \ref{fig:tree1}.
At any moment in time the agent will have some knowledge of the environment in the form of belief predicates, it will know in what operational state it is operating at the moment and it will have knowledge of what has been done in the past, still in the form of predicates. 
At any  particular time a plan can have multiple course of actions associated with it. Using both pre-computed and runtime simulation results, it would be possible to predict the one with the highest likelihood of success to determine whether or not the plan is likely to lead to the achievement of the current goals. To do this ``reward value'' can be applied to every applicable plan in $\Pi_t$. Let us define a ``reward function'' $f_\mathcal{R}$ that uses these principles to assigns a \emph{reward} value to every plan:

\begin{equation}
	\mathcal{R}_t : \Pi_t \rightarrow \mathbb{R^+}
\end{equation}

This association is clearly time dependant in the sense that the same plan might have a variable likelihood of success value when applied at different points of the mission under different environmental conditions. For this reason, at run time a ``reward update function''(which is defined at design time of the agent) needs to be evaluated each  time the agent is required to make a choice between different plans, and so different course of actions:

\begin{equation}
	\mathcal{RU} : \Pi_{t-1} \times \mathbb{R^+} \rightarrow \Pi_t \times \mathbb{R^+}
\end{equation}

The above runtime schema and method enables the agent to consider it options for the future. In the next section we are going to propose a way to assess its options with a probabilistic model checking technique and finally make the agent to decide about its next action (and likely consecutive actions).

\section{Plan Selection by Runtime Model Checking}


The applicable plan selection function $S_O$ makes a selection amongst plans that are triggered by the same event and that have a context that matches the current belief. The principles proposed in Section \ref{sec:principles} can be modelled with \gls{dtmc} as shown in Section \ref{sec:dtmcabstraction}.
However in a more complex plan selection mechanism, which may break the applicability of a \gls{dtmc} model, the $S_O$ block's plan selection can be replaced by a detailed plan-optimiser or by a combination of a simplified symbolic planner with options at its output and followed by detailed model checker for a time horizon of consequences of actions, which selects the best option. This is not unlike humans select their plans: first they consider a few alternatives and then check each for details. 
\begin{figure}[htbp]
\centering
\begin{tikzpicture}[auto,node distance=6mm,>=latex,font=\small]

	\tikzstyle{rect}=[thick,draw=black,inner sep=6,text centered,text width=140,minimum height=30,rounded corners=3]
	
	\node[rect] (n1) {Ranking of goals and selection the next most important goal G};
	\node[rect,below= of n1] (n2) {\emph{Using a symbolic planner}: a list of feasible symbolic plans is identified};
	\node[rect,below= of n2] (n3) {\emph{Using PRISM}: Probabilistic model checker is repeatedly applied to decide on the most likely best first action, which can be continued with plans likely to succeed to achieve G};
	
	\draw[->,thick] (n1) -- (n2);
	\draw[->,thick] (n2) -- (n3);

\end{tikzpicture}
\caption{The principle of selecting the best plan by methods of model checking. }
\label{runMC}
\end{figure}
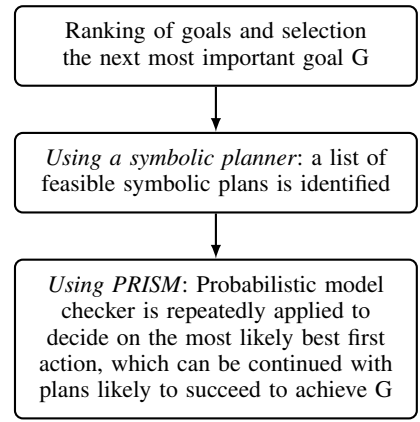
It is well known that model checking can often be used for plan
selection by generating counterexamples for the negation of a temporal
logic formula to be verified. On the other hand, planning rarely
answers queries about liveness and negative consequences are often hard to formulate as constraints, especially when indirectly implied, . For these reasons probabilistic model checking, in combination with planning in continuous and discrete space time, is a reasonable choice to check out feasible and most promising next actions and also consequent options for actions depending on how the future outcomes evolve. 

Fig. \ref{runMC} outlines the operation of the runtime model checker
for the selection of the next best action. PRISM can generate
counterexamples for queries \cite{HanKD09}, each of which is a trace in the
model. The counterexample generation produces the trace  with the
highest probability first. When we feed PRISM the negation of a query,
the first action in the trace, as a part of trace provided counter-example, is the next best action. 


In this scheme the agent decides on the next suitable goal using its reasoning rules in the top block of Fig. \ref{runMC}. Next a symbolic planner is applied to generate alternatives of plan and action sequences for the future. This list of symbolic plan/action sequences is then analysed by a probabilistic model checker to find the symbolic plan most likely to be successful. 

\section{Runtime model checking example in LISA}

Consider an \gls{asv} on an exploration mission. The vehicle has to
explore two areas, that we will call ``Area A'' and ``Area B''.

The system is trying to carry out two sets of actions:

\begin{figure}[h]
\centering
\begin{tikzpicture}[auto,node distance=2mm,>=latex]
 
	\tikzstyle{rect}=[thin,draw=black,dashed,inner sep=2mm,text width=30mm]

	\node[rect] (n1) {\emph{Explore Area A}
	\begin{itemize}
    	\item Go to Area A
    	\item Cover full area
   	\end{itemize}};
	\node[rect,right= of n1] (n2) {\emph{Explore Area B}
	\begin{itemize}
    	\item Go to Area B
    	\item Cover full area
   	\end{itemize}};

\end{tikzpicture}
\end{figure}
Each area is partitioned into a number of blocks: $A_1, A_2, \ldots, A_{N_A}$ for Area A and $B_1,B_2,\ldots,B_{N_B}$ for Area B. The exploration of the areas can be done in the same number of steps
accordingly, one block for each step. The vehicle consumes one
unit of fuel to explore one block. During the exploration, the weather
can change: with probability $p$, it becomes bad and with $1-p$ it
becomes good. When the former happens, the vehicle consumes twice as much fuel to explore each block. When the weather in the other area is
good, it can move to that area with  probability $1-q$ of a successful passage. When one area is
fully explored, the vehicle goes to the other area, assuming there is
still enough fuel. Moving between the two areas consumes one unit of
fuel, as well as going to each area from the base. At any point during
the mission, if the fuel tank becomes almost empty, e.g. when the fuel in
the tank is one unit, the vehicle has to go back to the
base. Once both areas are explored, the agent will head back to the
base. 

\begin{figure}[htbp]
\centering
\begin{tikzpicture}[auto,node distance=8mm,>=latex,font=\footnotesize,scale=0.75, every node/.style={transform shape}]

	\tikzstyle{rect}=[draw=black,inner sep=1mm,text centered,text width=25mm,minimum height=11mm]
	\tikzstyle{ell}=[draw=black,inner sep=1mm,text centered,text width=20mm,rounded corners=3.5mm,minimum height=11mm]

	\node[ell] (x0) {Initial choice};
	
	\node[rect,below left=8mm and -2mm of x0] (a1) {Go to Block $A_1$\\$Fuel:=Fuel-1$};
	\node[rect,below right=8mm and -2mm of x0] (a2) {Go to Block $B_1$\\$Fuel:=Fuel-1$};
	\draw[->] (x0) -- (a1);
	\draw[->] (x0) -- (a2);
	
	\node[ell,below left=8mm and -2mm of a1] (x1) {In Block $A_1$\\$Fuel==0$};
	\node[ell,below right=8mm and -1mm of a1] (x2) {In Block $A_1$\\$Fuel>0$};
	\draw[->] (a1) -- (x1);
	\draw[->] (a1) -- (x2);
	
	\node[rect,below= of x1,minimum height=8mm,text width=20mm] (a3) {Back to base};
	\node[ell,below left=8mm and -2mm of x2] (x3) {Good weather\\in $A_1$};
	\node[ell,below right=8mm and 0mm of x2] (x4) {Bad weather\\in $A_1$};
	\draw[->] (x1) -- (a3);
	\draw[->] (x2) to node[text width=4mm,above left=-1mm and 0mm]{1-p} (x3);
	\draw[->] (x2) to node[text width=2mm,above right=-1mm and 0mm]{p} (x4);
	
	\node[ell,below= of a3] (x5) {In Base};
	\node[rect,below=of x3,text width=26mm] (a4) {Explore Block $A_1$\\Go to Block $A_2$\\$Fuel:=Fuel-1$};
	\node[ell,below left=8mm and -8mm of x4] (x6) {Bad weather in Block $B_1$};
	\node[ell,below right=8mm and -8mm of x4] (x7) {Good weather\\in block $B_1$};
	\draw[->] (a3) -- (x5);
	\draw[->] (x3) -- (a4);
	\draw[->] (x4) to node[text width=2mm,above left=-1mm and 0mm]{p} (x6);
	\draw[->] (x4) to node[text width=4mm,above right=-1mm and 0mm]{1-p} (x7);
	
	\node[rect,below=of x6,text width=26mm] (a5) {Explore Block $A_1$\\Go to Block $A_2$\\$Fuel:=Fuel-2$};
	\node[rect,below right=8mm and -15mm of x7] (a7) {Go to Block $B_1$\\$Fuel:=Fuel-1$};
	\draw[->] (x6) -- (a5);
	\draw[->, dotted] (x7) to node[text width=4mm,above left=-1mm and 0mm]{1-q} (a5);
	\draw[->, dotted] (x7) to node[text width=2mm,above right=-1mm and 0mm]{q} (a7);
	
	\node[ell,below= of a5] (x8) {In Block $A_2$};
	\node[ell,below= of a7] (x9) {In Block $B_1$};
	\draw[->] (a4) [out=-90,in=180] to (x8);
	\draw[->] (a5) -- (x8);
	\draw[->] (a7) -- (x9);

\end{tikzpicture}
\caption{The partial Abstract \gls{dtmc} for the example}\label{fig:example}
\label{fig:exmp1}
\end{figure}
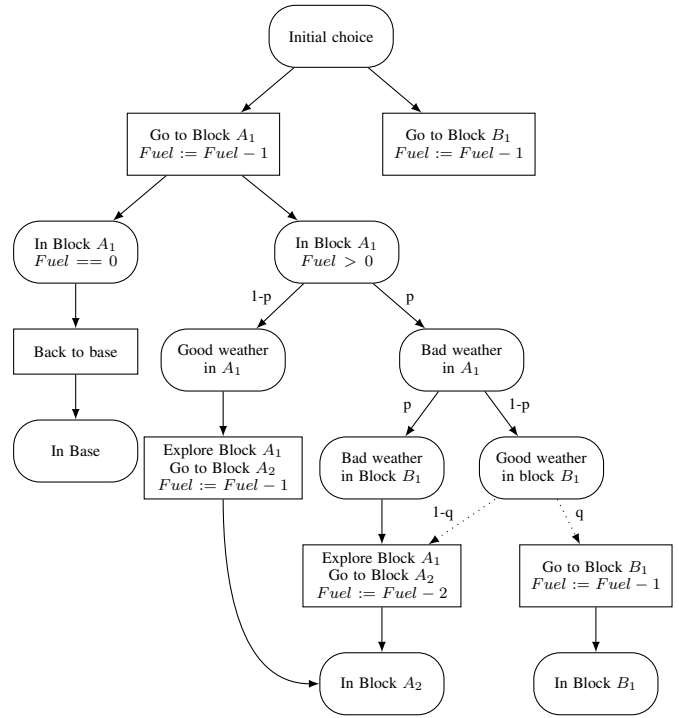

The symbolic plans are simply extracted from pre-defined plans in LISA by an action of the LISA agent, which generates symbolic plans to satisfy pre and post conditions, including  operational state history. The probabilities for the \gls{dtmc} in PRISM are obtained from the action feedback 
probabilities in the LISA program. Hence the programming framework also lends itself
to automated extraction of PRISM models of runtime verification for
decision making.
Fig. \ref{fig:exmp1} illustrates a partial graph of the \gls{dtmc} that models the
example for the design time verification. The dotted lines represent
the plans that need to be chosen. The reward for each plan can be the
probability of successfully fulfilling the mission from the
plan, which can be specified as a reachability query. This probability
is then computed by PRISM for each state in the DTMC. The reward for
the plan ``Explore Block $A_1$ and go to Block $A_2$'' is the
probability at the state ``In Block $A_2$'', while the reward for the
other plan is the probability at the state ``In Block $B_1$''.
The Appendix contains the sample PRISM program for Fig. \ref{fig:exmp1}.


\section{Implementation of LISA}

A suitable framework to implement LISA is MOOS-IvP \cite{newman2003,moosivpwebsite,benjamin2012}. MOOS is a inter-process communications-middleware software that is structured in a star-like fashion. It features a central node called the MOOS Database (MOOSDB) and a set of applications that communicate with each other through the MOOSDB in a publish/subscribe manner. IvP is a MOOS application that optimises behaviour selection in actions of the agent. Behaviours for action execution are modules of the IvP that \emph{compete} over the definition of control values that will be used to act on the environment. For example for control variables such as ``direction'' and ``speed'', behaviours can be ``waypoint following'', ``collision avoidance'' and so on.
This method is ideal to reduce complexity of the reasoning as it allows to further abstract action definitions and it significantly reduces the need for nested conditional statements within symbolic plans.

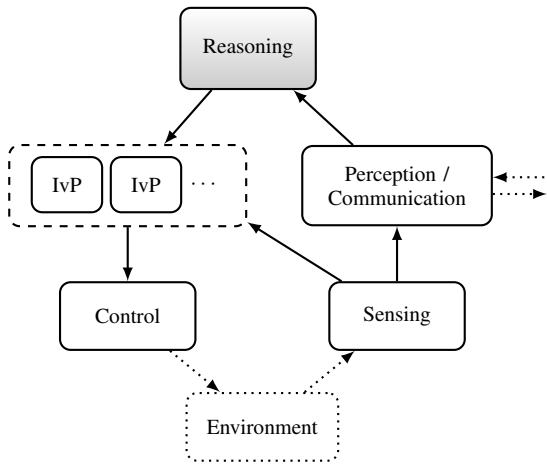
\begin{figure}[htbp]
\centering
\begin{tikzpicture}[auto,node distance=8mm,>=latex,font=\small,scale=0.9, every node/.style={transform shape}]

	\tikzstyle{rc}=[thick,draw=black,rounded corners,inner sep=3mm,minimum height=10mm,minimum width=20mm]
	
	\node (reasoning) [rc,left color=white,right color=white!80!black,shading angle=0,minimum height=12mm]  {Reasoning};
	
	\node (ivpbox) [rc,below left=8mm and -10mm of reasoning,dashed,minimum width=35mm,minimum height=12mm]  {};
	\node (ivp1) [rc,below left=-10.5mm and -14mm of ivpbox,minimum height=8mm,minimum width=10mm]  {IvP};
	\node (ivp2) [rc,right=1mm of ivp1,minimum height=8mm,minimum width=10mm]  {IvP};
	\node (dots) [right=0mm of ivp2] {$\cdots$};
	
	\node (perc) [rc,right= of ivpbox,text width=22mm,text centered] {Perception / Communication};
	\node (con) [rc,below= of ivpbox] {Control};
	\node (sens) at (con -| perc)[rc] {Sensing};
	\node (ref1) at ($(con)!0.5!(sens)$) []{};
	\node (env) [rc,dotted,below=10mm of ref1] {Environment};
	
	\draw[->,thick] (reasoning) -- (ivpbox);
	\draw[->,thick] (perc) -- (reasoning);
	\draw[<-,thick,dotted] (perc.5) -- +(8mm,0);
	\draw[->,thick,dotted] (perc.-5) -- +(8mm,0);
	\draw[->,thick] (ivpbox) -- (con);
	\draw[->,thick] (sens) -- (ivpbox.-17);
	\draw[->,thick] (sens) -- (perc);
	\draw[->,thick,dotted] (con) -- (env);
	\draw[->,thick,dotted] (env) -- (sens);

\end{tikzpicture}
	\caption{Implementation of the LISA architecture using IvP planning engines. Every node contains multiple sub-nodes of lower level skills of the LISA agent.}
	\label{fig:imp}
\end{figure}

 Fig. \ref{fig:imp} is illustrates the overall structure of the system. Each block is a collection of ``skills'' of the agent that can be implemented as MOOS processes (nodes). In MOOS the communication between nodes happens through the central database. The \emph{Reasoning} block performs symbolic planning and can activate IvP modules to optimise action execution within  plans of the LISA architecture. The \emph{Perception/Communication} block converts filtered sensing data and the communication channels into enriched symbolic information for the agent Reasoning. The IvP modules operate an intermediate feedback loop with some processes within the Sensing and Control blocks through their variables. Processes in the  \emph{Control} block interact with physical actuators that operate in the environment. Processes in the\emph{Sensing} block interact with physical sensors and take care of filtering data in order to minimise noise.

Behaviours are internal modules of the IvP application that reproduce a particular action over a set of control variables $c_1,c_2,\ldots,c_n$, and generate at every cycle a piecewise linear function called ``IvP function'' $f(c_1,c_2,\ldots,c_n)$ that maps points of the decision space to values that reflect the degree to which that control array supports the action. Once these functions are produced a multi-objective optimization problem is solved by another internal module called the IvP-solver:
\begin{equation}
\begin{array}{rl}
	\arg \underset{c_1,\ldots,c_n}{\max} & w_1 f_1(c_1,\ldots,c_n) + \cdots + w_k f_k(c_1,\ldots,c_n)\\
	\text{s.t.} & f_i \; \text{is an IvP piecewise defined function}\\
		& w_i \in \mathbb{R}_{\geq 0}
\end{array}
\end{equation}
where $w_1,w_2,\ldots,w_n$ are called \emph{priority weightings}. In MOOS-IvP the priority weightings are influenced by two main factors:
\begin{enumerate}
\item
	With every behaviour is associated a set of binary flags that allow control over the activation time of the behaviour itself. These flags can be conditionally modified by the behaviour itself, which means that the behaviour has partial control over its own state, and can also be associated with external variables that can be modified by other nodes.
\item
	A hierarchical mode system is defined within IvP that allows to organise the behaviour activation according to declared mission modes. Modes and sub-modes can be declared in line with the designer's own concept of mission evolution, and conditional statements can be implemented so to switch between modes. Modes can also be associated with external variables that can be modified by other nodes.
\end{enumerate}

The control over the value of the priority weightings allow the Reasoning of LISA an appropriate level of control over the functionality of the behaviours:
\begin{itemize}
\item
	Actions $a_1,a_2,\ldots,a_{n_a}$ from the agent set $A \subset \mathcal{F}\setminus B$ can directly activate or deactivate behaviours. At the end of every reasoning cycle (See Fig. \ref{fig:lisacycle}), LISA takes a set of next actions from the available plans in the Intentions set and issues them for execution. These actions can be external or internal routines. A way to integrate LISA with MOOS-IvP is to create external actions that can modify variables that are linked to the behaviour's activation flags. 
	For example an action of the type ``Go to area A'' can be executed in many different ways depending on the state of the environment, on the distance and so on. For this purpose the reasoning can send an activation flag to several behaviours that all operate on the same space (for example ``direction'' and ``speed'') and that compete with each other.
\item
	At any given moment in time, a set $X_t$ of operational states from $X=\{x_1,x_2,\ldots,x_{n_x}\}$, defined in Eq. (\ref{eq:opstates}), are active in the LISA agent. These operational states are updated and activated according to the mission status. When operational state are set or modified in LISA external actions can be executed to modify variables associated with the hierarchical mode structure of IvP. For example an operational state of the type ``Exploring area A'' can be associated with a set of behaviours that include for example \emph{area covering}, but at the same time \emph{collision avoidance} and \emph{SLAM}.
\end{itemize}


\section{Conclusions}

The history of   methods for controlling of autonomous
vehicles has been briefly reviewed and  a new agent architecture called Limited
Instruction Set Agent (LISA)  has been proposed. The architecture originates from previous AgentSpeak implementations such as Jason and Jade. By reviewing in detail  the reasoning cycle and the implementation of Jason, we identified several design features which make verification complex  and proposed an alternative
architecture that improves verifiability of agent reasoning in physical environments. In particular, we reduced
complexity by simplifying the reasoning cycle with the use of
multi-threaded processing and proposed the use of model checking
techniques on two levels: (1) \emph{Design-time model checking}, we
prove that it is possible to abstract the agent to a \gls{dtmc} and in
turn verify the decision making process with existing model checking
techniques, (2) \emph{Runtime verification of executable plan
  selection}, the agent is able to assess a tree of possible future
outcomes and select a plan which is most likely to succeed in reaching
the mission goals.

\section*{Appendix}

Example PRISM program of runtime verification in the LISA system. 

\begin{lstlisting}[language=Prism,numbers=left]
dtmc

const int No = 15; 
const int Na = 5;
const int Nb = 5;

const double Pa = 0.1; // probability of bad weather in Area A
const double Pb = Pa;  // probability of bad weather in Area B
const double Pi = 0.5; // probability of initial choice between Area A and Area B
const double Ps = 0.6; // probability of switch between Area A and Area B in case of bad weather

module robot1
	a1 : [0..Na] init 0; // Area A
  b1 : [0..Nb] init 0; // Area B
  oil : [0..No] init No; // oil level
  s  : [0..3] init 0; // 0: base, 1: Area A, 2: Area B, 3: mission aborted

  // initial choice
	[] (s = 0) & (a1 = 0)& (oil > 0) 
	-> Pi: (s'=1)&(oil'=oil-1) + (1-Pi): (s'=2)&(oil'=oil-1);

  // decision in Area A
[tick2] (s = 1) & (t = 0) & (a1 < Na) & 
(oil > 0) & (w1 = 0) -> (a1'=a1+1) & (oil'=oil-1);
[tick2] (s = 1) & (t = 0) & (a1 = Na) 
& (b1 < Nb) & (oil > 0) -> (s'=2) & (oil'=oil-1);
[tick2] (s = 1) & (t = 0) & (a1 < Na) & (b1 < Nb) 
	& (oil > 1) & (w1 = 1) & (w2 = 1) 
	        -> (a1'=a1+1) & (oil'=oil-2);
[tick2] (s = 1) & (t = 0) & (a1 < Na)  & (b1 < Nb)
	& (oil > 1) & (w1 = 1) & (w2 = 0) 
	  -> Ps: (a1'=a1+1) & (oil'=oil-2) + 
			(1-Ps): (s'=2) & (oil'=oil-1);

  // decision in Area B
[tick2] (s = 2) & (t = 0) & (b1 < Nb) & (oil > 0) 
& (w2 = 0) -> (b1'=b1+1) & (oil'=oil-1);
[tick2] (s = 2) & (t = 0) & (b1 = Nb) & (a1 < Na) 
& (oil > 0) -> (s'=1) & (oil'=oil-1);
[tick2] (s = 2) & (t = 0) & (b1 < Nb) & (a1 < Na) 
& (oil > 1) & (w2 = 1) & (w1 = 1) 
	     -> (b1'=b1+1) & (oil'=oil-2);
[tick2] (s = 2) & (t = 0) & (b1 < Nb) & (a1 < Na) 
& (oil > 1) & (w2 = 1) & (w1 = 0) 
	   -> Ps: (b1'=b1+1) & (oil'=oil-2) + (1-Ps): (s'=2) & (oil'=oil-1);
 [tick2] (s > 0) & (s < 3) & (oil = 0) -> (s'=3);

endmodule

module environment
  t: [0..1] init 1; // control weather change
  
  [tick1] (s > 0) & (s < 3) & (t = 1) -> (t' = 0);  
  [tick2] (s > 0) & (s < 3) & (t = 0) -> (t' = 1);  
endmodule

module weather1
	w1 : [0..1];

  [tick1] (s > 0) & (s < 3) & (t = 1) -> Pa:(w1' = 1) + (1-Pa):(w1'=0);
endmodule

module weather2
	w2 : [0..1];

  [tick1] (s > 0) & (s < 3) & (t = 1) -> Pb:(w2' = 1) + (1-Pb):(w2'=0);
endmodule
\end{lstlisting}

%
%



\bibliographystyle{IEEEtran}
\bibliography{IEEEabrv,references,bibliog}

\end{document}

%% file: include.tex
\usepackage{cite}

\usepackage[pdftex]{graphicx}
\usepackage[font={small,it}]{caption} 
\usepackage{tikz}
\usetikzlibrary{positioning,shapes,arrows,calc}

\usepackage{fixltx2e} 

\usepackage{color}
\usepackage{courier} 

\usepackage[cmex10]{amsmath}
\usepackage{amsfonts} 
\usepackage{amssymb} 

\usepackage{amsthm} 

	\theoremstyle{definition} 

	\newtheorem{theorem}{Theorem}

\usepackage{algorithmic} 

\usepackage{listings} 

	\definecolor{ashgrey}{rgb}{0.5, 0.5, 0.5}

	\lstdefinelanguage{Prism}{ 
        basicstyle=\color{black}\tiny\ttfamily, 
        keywords= {bool, C, ceil, const, ctmc, double, dtmc, endinit, endmodule, endrewards, endsystem, F, false, floor, formula, G, global, I, init, int, label, max, mdp, min, module, nondeterministic, P, Pmin, Pmax, prob, probabilistic, R, rate, rewards, Rmin, Rmax, S, stochastic, system, true, U, X}, 
        keywordstyle={\bfseries\color{black}}, 
        numberstyle=\tiny\color{black}, 
        comment=[l] {//}, morecomment=[s]{/*}{*/}, 
        commentstyle= \color{ashgrey}, 
        tabsize=4, 
        captionpos=b, 
        escapechar=@ 
}

\usepackage{url}
\usepackage{hyperref} 
    \hypersetup{
        colorlinks,
        citecolor=black,
        filecolor=black,
        linkcolor=black,
        urlcolor=black
    }

\usepackage{array} 
\usepackage{supertabular}

\usepackage{glossaries}
	\glsdisablehyper

    \newacronym{agv}{AGV}{Autonomous Ground Vehicle}
    \newacronym{aop}{AOP}{Agent Oriented Programming}
    \newacronym{asv}{ASV}{Autonomous Surface Vehicle}
    \newacronym{auv}{AUV}{Autonomous Underwater Vehicle}
    
    \newacronym{bdi}{BDI}{Belief-Desire-Intention}
    
    \newacronym{cat}{CAT}{Cognitive Agent Toolbox}
    \newacronym{cdt}{CDT}{Constrained Delaunay Triangulation}
    \newacronym{cog}{COG}{Centre Of Gravity}
    
    \newacronym{dof}{DOF}{Degrees Of Freedom}
    \newacronym{dtmc}{DTMC}{Discrete-Time Markov Chain}
    
    \newacronym{ekf}{EKF}{Extended Kalman Filter}
    
    \newacronym{gps}{GPS}{Global Positioning System}
    \newacronym{gui}{GUI}{Graphical User Interface}
    
    \newacronym{hs}{HS}{Hybrid System}
    
    \newacronym{imu}{IMU}{Inertial Measurement Unit}
    \newacronym{ivp}{IvP}{Interval Programming}
    
    \newacronym{lct}{LCT}{Local Clearance Triangulation}
    
    \newacronym{mcmas}{MCMAS}{Model Checker for Multi-Agent Systems}
    \newacronym{mol}{MOL}{Machine Ontology Language}
    \newacronym{moos}{MOOS}{Mission Oriented Operating Suite}
    
    \newacronym{nasa}{NASA}{National Aeronautics and Space Administration}
    \newacronym{nlp}{NLP}{Natural Language Programming}
    
    \newacronym{oop}{OOP}{Object Oriented Programming}
    
    \newacronym{pact}{PACT}{Pilot Authority and Control of Tasks}
    \newacronym{prs}{PRS}{Procedural Reasoning System}
    
    \newacronym{rrt}{RRT}{Rapidly-exploring Random Tree}
    
    \newacronym{sis}{SIS}{Sequential Importance Sampling}
    \newacronym{slam}{SLAM}{Simultaneous Localization And Mapping}
    \newacronym{smcp}{SMCP}{Standard Marine Communication Phrases}
    
    \newacronym{ukf}{UKF}{Unscented Kalman Filter}
    \newacronym{usv}{USV}{Unmanned Surface Vehicle}
    \newacronym{uav}{UAV}{Unmanned Aerial Vehicle}
    \newacronym{uuv}{UUV}{Unmanned Underwater Vehicle}
    
    \newacronym{vr}{VR}{Virtual Reality}